\documentclass[sigconf]{aamas} 

\usepackage{balance} 

\usepackage{graphicx}
\usepackage{amsmath}
\usepackage{amsthm}
\usepackage{booktabs}
\usepackage{algorithm}
\usepackage[noend]{algorithmic}
\usepackage{tikz}
\usepackage{mathtools}
\usepackage{centernot}
\usepackage{enumerate}
\usepackage{cleveref}
\usepackage{subfig}
\usepackage{todonotes}

\usepackage{booktabs}
\usepackage{multirow}
\usepackage{makecell}

\usetikzlibrary{arrows}

\usetikzlibrary{patterns}

\makeatletter
\DeclareRobustCommand*\cal{\@fontswitch\relax\mathcal}
\makeatother
\def\hb{\hbox to 10.7 cm{}}
\newtheoremstyle{break}
  {\topsep}{\topsep}%
  {\itshape}{}%
  {\bfseries}{}%
  {\newline}{}%
\theoremstyle{break}
\usepackage[inline, shortlabels]{enumitem}  
\setlist{leftmargin=*}

\newtheorem{definition}{Definition}
\newtheorem{principle}{Principle}
\newtheorem{example}{Example}

\usepackage{thmtools}
\usepackage{thm-restate}

\newtheorem{theorem}{Theorem}[section]
\newtheorem{proposition}[theorem]{Proposition}

\usepackage[normalem]{ulem}

\makeatletter
\renewcommand{\p@subfigure}{}
\makeatother

\def\QBAG{\mbox{\ensuremath{(\Args, \is, \att, \supp)}}}

\def\Args{\ensuremath{\mathit{Args}}} 
\def\Att{\ensuremath{\mathit{Att}}}
\def\Supp{\ensuremath{\mathit{Supp}}}
\def\is{\ensuremath{\tau}}
\def\fs{\ensuremath{\sigma}}
\def\att{\Att}
\def\supp{\Supp}
\def\ddom{\ensuremath{\mathrm{ddom}}}

\def\graph{\ensuremath{\mathcal{G}}} 
 
\def\interval{\ensuremath{\mathbb{I}}}

\def\arga{\ensuremath{\mathsf{a}}}
\def\argb{\ensuremath{\mathsf{b}}}
\def\argc{\ensuremath{\mathsf{c}}}
\def\argd{\ensuremath{\mathsf{d}}}
\def\arge{\ensuremath{\mathsf{e}}}

\def\argx{\ensuremath{\mathsf{x}}}
\def\argy{\ensuremath{\mathsf{y}}}
\def\argz{\ensuremath{\mathsf{z}}}

\newcommand{\argnode}[3]{\mbox{\ensuremath{#1~(#2)\!:\!\mathbf{#3}}}}
\newcommand{\pargnode}[1]{\mbox{\ensuremath{#1~\phantom{(e)\!:\!\mathbf{e}}}}}

\tikzset{
noanode/.style={scale=0.55,dashed, circle, draw=black!60, minimum size=10mm, font=\bfseries},
unanode/.style={scale=0.55,circle, draw=black!75, minimum size=10mm, font=\bfseries},
xunanode/.style={scale=0.55,circle, draw=black!75, minimum size=10mm, font=\bfseries, line width=2.5pt},
invnode/.style={scale=0.55,circle, draw=white!0, minimum size=0mm, font=\bfseries},
anode/.style={scale=0.55,circle, draw=lightgray!60, minimum size=10mm, font=\bfseries},
xanode/.style={scale=0.55, circle, fill=lightgray, draw, minimum size=10mm, font=\bfseries,line width=2.5pt},
xnoanode/.style={scale=0.55, circle, dashed, draw, minimum size=10mm, font=\bfseries,line width=2.5pt},
sxnoanode/.style={scale=0.55, circle, dotted, draw, minimum size=10mm, font=\bfseries, line width=0.85pt}
}

\setcopyright{ifaamas}
\acmConference[AAMAS '26]{Proc.\@ of the 25th International Conference
on Autonomous Agents and Multiagent Systems (AAMAS 2026)}{May 25 -- 29, 2026}
{Paphos, Cyprus}{C.~Amato, L.~Dennis, V.~Mascardi, J.~Thangarajah (eds.)}
\copyrightyear{2026}
\acmYear{2026}
\acmDOI{}
\acmPrice{}
\acmISBN{}

\acmSubmissionID{762}

\title[Strength Change Explanations in Quantitative Argumentation]{Strength Change Explanations in Quantitative Argumentation} 

\author{Timotheus Kampik}
\affiliation{
  \institution{Umeå University}
  \city{Umeå}
  \country{Sweden}}
\email{tkampik@cs.umu.se}

\author{Xiang Yin}
\affiliation{
  \institution{Imperial College London}
  \city{London}
  \country{United Kingdom}}
\email{x.yin20@imperial.ac.uk}

\author{Nico Potyka}
\affiliation{
  \institution{Cardiff University}
  \city{Cardiff}
  \country{United Kingdom}}
\email{PotykaN@cardiff.ac.uk}

\author{Francesca Toni}
\affiliation{
  \institution{Imperial College London}
  \city{London}
  \country{United Kingdom}}
\email{f.toni@imperial.ac.uk}

\begin{abstract}
In order to make argumentation-based inference contestable, it is crucial to explain what changes can achieve a desired (instead of the contested) inference result.
To this end, we introduce \emph{strength change explanations} for quantitative (bipolar) argumentation graphs.
Strength change explanations describe changes to the initial strengths of a subset of the arguments in a given graph that can achieve a desired ordering based on the final strengths of some (potentially different) subset of arguments.
We show that the existing notions of \emph{inverse} and \emph{counterfactual} problems can be reduced to strength change explanations.
We also prove basic soundness and completeness properties of our strength change explanations, and demonstrate their existence and non-existence in some special cases.
By applying a heuristic search, we demonstrate that we can often successfully find strength change explanations for layered graphs that are common in typical application scenarios; still, limitations remain for settings where we do not provide guarantees for the presence (or absence) of explanations.
\end{abstract}


\newcommand{\BibTeX}{\rm B\kern-.05em{\sc i\kern-.025em b}\kern-.08em\TeX}

\begin{document}


\pagestyle{fancy}
\fancyhead{}


\maketitle 


\section{Introduction}
\label{sec:intro}
In order to facilitate human-centricity, applications of Artificial Intelligence (AI)
need to be \emph{contestable}: not only must machines explain the results of their decision-making processes to humans; in addition, humans must be able to challenge these results~\cite{10.1145/3449180}.
Computational argumentation, in which inferences are drawn from potentially dynamic graphs modelling arguments (nodes), as well as attack and support relationships between them (edges), may have the potential to be a key enabler of contestable AI~\cite{DBLP:journals/corr/abs-2405-10729,DBLP:journals/frai/DietzKM22}.

One way to achieve contestability is to enable machines to provide, given a decision outcome, explanations in the sense of sets of required changes that lead to a more desirable outcome~\cite{guidotti2024counterfactual,Stepin.et.al:2021}.
We define such explanations for Quantitative Bipolar Argumentation Graphs (QBAGs), graphs with weighted nodes and directed edges representing two binary relations modelling \emph{attack} and \emph{support}, respectively.
\emph{Gradual semantics} then draw inferences from QBAGs by updating the weights from \emph{initial strengths} to \emph{final strengths}, given the graph topology of the QBAG.
QBAGs play an important role in argumentative eXplainable AI (XAI), a line of research that aims to advance the study and application of computational argumentation in the broader 
XAI context~\cite{Cyras.et.al:2021-IJCAI}.
A range of works showcases the application of QBAGs to XAI use cases, such as explainable image recognition~\cite{ayoobi2024argumentative} and recommendation systems~\cite{RAGO2021103506}.
To facilitate contestability in QBAGs, we study which changes to the initial strengths of a subset of the arguments in a QBAG can yield a desired outcome in terms of the ordering arising from the final strengths of another subset of the QBAG's arguments.

For example, assume the arguments $\arga$ and $\arge$ that model variables of a credit application decision in an intermediate layer of a layered QBAG; $\argd$ models a variable influencing both $\arga$ and $\arge$. Finally, $\argb$ and $\argc$ in the output layer model the acceptance of the application only if the final strength of $\argc$ is greater than the final strength of $\argb$. This ordering can potentially be affected by changes to the initial strengths of $\arga$ and $\arge$.
We may want to identify such initial strength changes that specifically affect a change of the final strength ordering $\argb \succ \argc$ ($\argb$'s final strength is greater than $\argc$'s, \emph{application rejected}) to $\argc \succ \argb$ (\emph{application accepted}), cf. Figure~\ref{fig:intro}.

We can identify such changes, which we call \emph{Strength change eXplanations} (SXs), by generalizing the so-called \emph{inverse problem} in quantitative argumentation, which describes the assignment of initial strengths to all arguments in a QBAG such that a desired final strength ordering of these arguments is achieved~\cite{DBLP:conf/ijcai/OrenYVB22}.

Before we commence the formal part of this paper, let us expand on the colloquial intuition of an SX. An SX depends on a QBAG, a gradual semantics, and a subset of the QBAG's arguments, which we call \emph{mutable arguments}. It defines a set of (mutable argument, initial strength)-tuples that, if applied to the QBAG, yield a specific desired ordering given by the final strengths of (some of) the arguments in the QBAG.
Roughly, we say that an SX is \emph{$\epsilon$-approximate} if the desired ordering cannot be achieved by a substantially better SX 
in terms of the sum of all changes to arguments' initial strengths (a smaller sum is better).
Intuitively, a \emph{$0$-approximate} SX is \emph{optimal}, while a larger $\epsilon$ indicates weaker optimality guarantees.
Below, we give an example of SXs, applying a simplistic gradual semantics that (given an acyclic QBAG) traverses the graph in topological order and, given an argument, subtracts the final strengths of all attackers from the argument's initial strength, while adding the final strengths of all supporters\footnote{We use this semantics so that readers can easily verify the examples; our results do not depend on it.}.

\begin{example}
\label{ex:intro}
Consider the QBAG in Figure~\ref{fig:intro:g}. Nodes in the graphs are arguments, $\argnode{\argx}{i}{f}$ represents argument $\argx$ with initial strength $\is(\argx) = i$ and final strength $\fs(\argx) = \mathbf{f}$, and edges labelled $+$ and $-$ represent support and attack, respectively. The final strength of $\argb$ is greater than the final strength of $\argc$: $\fs(\argb) > \fs(\argc)$. We want to find changes to the initial strengths of the arguments $\arga$ and $\arge$ that yield  $\fs(\argb) < \fs(\argc)$.
Such changes are applied in Figures~\ref{fig:intro:gprime}, \ref{fig:intro:gprimeprime}, and \ref{fig:intro:gstar}. The changes applied in Figures~\ref{fig:intro:gprime} and \ref{fig:intro:gprimeprime} are $\epsilon$-approximate, given $\epsilon = 1$ (i.e., technically any $\epsilon \geq 1$ would work as well).
Clearly, the optimal way of achieving the desired ordering is increasing the initial strength of $\arga$ by marginally more than $1$.
As the changes applied in $\graph'$ and $\graph''$ are $|\tau_{\graph'}(\arga) - \tau_{\graph}(\arga)| + |\tau_{\graph'}(\arge) - \tau_{\graph}(\arge)| = 2$ and $|\tau_{\graph''}(\arga) - \tau_{\graph}(\arga)| = 2$, respectively, the changes are still within the approximation ``wiggle room'' of $1$.
The changes applied to $\graph^*$ are not $\epsilon$-approximate given $\epsilon = 1$: we have increased the initial strength of $\arga$ by $1$ and of $\arge$ by $2$, but we could have increased the initial strength of $\arge$ by $< 2 - 1$ (e.g., by just $0.5$) and still achieve the desired ordering.
\begin{figure}[!ht]
    \subfloat[$\graph$]{
    \label{fig:intro:g}
    \begin{tikzpicture}[scale=0.55]
        \node[unanode]    (a)    at(0,2)  {\argnode{\arga}{1}{2}};
        \node[unanode]  (b)    at(0,0)  {\argnode{\argb}{8}{6}};
        \node[unanode]    (c)    at(2,0)  {\argnode{\argc}{1}{4}};
        \node[unanode]    (d)    at(1,4)  {\argnode{\argd}{1}{1}};
        \node[unanode]    (e)    at(3,2)  {\argnode{\arge}{2}{1}};
        
        \path [->, line width=0.5mm]  (a) edge node[left] {-} (b);
        \path [->, line width=0.5mm]  (a) edge node[above] {+} (c);
        \path [->, line width=0.5mm]  (e) edge node[right] {+} (c);
        \path [->, line width=0.5mm]  (d) edge node[left] {+} (a);
        \path [->, line width=0.5mm]  (d) edge node[right] {-} (e);
    \end{tikzpicture}
    }
    \hspace{20pt}
    \centering
    \subfloat[$\graph'$]{
    \label{fig:intro:gprime}
    \begin{tikzpicture}[scale=0.55]
        \node[xnoanode]    (a)    at(0,2)  {\argnode{\arga}{2}{3}};
        \node[unanode]  (b)    at(0,0)  {\argnode{\argb}{8}{5}};
        \node[unanode]    (c)    at(2,0)  {\argnode{\argc}{1}{6}};
        \node[unanode]    (d)    at(1,4)  {\argnode{\argd}{1}{1}};
        \node[xnoanode]    (e)    at(3,2)  {\argnode{\arge}{3}{2}};
        
        \path [->, line width=0.5mm]  (a) edge node[left] {-} (b);
        \path [->, line width=0.5mm]  (a) edge node[above] {+} (c);
        \path [->, line width=0.5mm]  (e) edge node[right] {+} (c);
        \path [->, line width=0.5mm]  (d) edge node[left] {+} (a);
        \path [->, line width=0.5mm]  (d) edge node[right] {-} (e);
    \end{tikzpicture}
    } \\
    \centering
    \subfloat[$\graph''$]{
    \label{fig:intro:gprimeprime}
    \begin{tikzpicture}[scale=0.55]
        \node[xnoanode]    (a)    at(0,2)  {\argnode{\arga}{3}{4}};
        \node[unanode]  (b)    at(0,0)  {\argnode{\argb}{8}{4}};
        \node[unanode]    (c)    at(2,0)  {\argnode{\argc}{1}{6}};
        \node[unanode]    (d)    at(1,4)  {\argnode{\argd}{1}{1}};
        \node[unanode]    (e)    at(3,2)  {\argnode{\arge}{2}{1}};
        
        \path [->, line width=0.5mm]  (a) edge node[left] {-} (b);
        \path [->, line width=0.5mm]  (a) edge node[above] {+} (c);
        \path [->, line width=0.5mm]  (e) edge node[right] {+} (c);
        \path [->, line width=0.5mm]  (d) edge node[left] {+} (a);
        \path [->, line width=0.5mm]  (d) edge node[right] {-} (e);
    \end{tikzpicture}
    }
    \hspace{20pt}
    \centering
    \subfloat[$\graph^*$]{
    \label{fig:intro:gstar}
    \begin{tikzpicture}[scale=0.55]
        \node[xunanode]    (a)    at(0,2)  {\argnode{\arga}{2}{3}};
        \node[unanode]  (b)    at(0,0)  {\argnode{\argb}{8}{5}};
        \node[unanode]    (c)    at(2,0)  {\argnode{\argc}{1}{7}};
        \node[unanode]    (d)    at(1,4)  {\argnode{\argd}{1}{1}};
        \node[xunanode]    (e)    at(3,2)  {\argnode{\arge}{4}{3}};
        
        \path [->, line width=0.5mm]  (a) edge node[left] {-} (b);
        \path [->, line width=0.5mm]  (a) edge node[above] {+} (c);
        \path [->, line width=0.5mm]  (e) edge node[right] {+} (c);
        \path [->, line width=0.5mm]  (d) edge node[left] {+} (a);
        \path [->, line width=0.5mm]  (d) edge node[right] {-} (e);
    \end{tikzpicture}
    } 
\caption{QBAG $\graph$ and its updates $\graph'$, $\graph''$, and $\graph^*$. Here and henceforth, a node labelled $\argnode{\argx}{i}{f}$ represents argument $\argx$ with initial strength $\is(\argx) = i$ and final strength $\fs(\argx) = \mathbf{f}$. Edges labelled $+$ and $-$ represent support and attack, respectively. Arguments with bold borders are strength change explanation arguments, given the desired ordering $\langle \argc, \argb \rangle$ and the \emph{mutable set} $\{\arga, \arge\}$; arguments with bold dashed borders make up \emph{$1$-approximate} strength change explanations.}
\label{fig:intro}
\end{figure}
\end{example}
Taking the idea sketched above as a starting point, the contributions of this paper are the following:
\begin{enumerate*}[label=(\roman*)]
\item We provide a formal framework for SXs (Section~\ref{sec:explanations});
\item We analyse basic properties of \emph{optimal} SXs (Section~\ref{sec:optimality});
\item We demonstrate existence and non-existence guarantees of SXs for some cases (Section~\ref{sec:approximating});
\item We implement a heuristics-based search for identifying strength change explanations, drawing from what we have learned from the theoretical analysis and empirically demonstrate the feasibility of finding SXs in small layered QBAGs, as well as some limitations (Section~\ref{sec:application});
\item Finally, we formally establish the connection between SXs related approaches in the literature (Section~\ref{sec:characterizations}).
\end{enumerate*}

\section{Related Work}
\label{sec:related}
The work we present in this paper extends the line of research on argumentative XAI~\cite{Cyras.et.al:2021-IJCAI,vassiliades_bassiliades_patkos_2021} and, more specifically, on the study of explaining inferences drawn from argumentation frameworks---QBAGs, in our case.
QBAG explainability has recently been studied in a range of works~\cite{KAMPIK2023109066,DBLP:journals/ijar/KampikPYCT24,yin2024qarg,DBLP:journals/corr/abs-2407-08302}.
Most relevant in our context are studies that define and analyse
\begin{enumerate*}[label=(\roman*)]
    \item (any, not only initial strength) changes to QBAGs that can explain changes in the final strength-based ordering of two arguments~\cite{KAMPIK2023109066};
    \item initial strength changes to the arguments in a QBAG that can affect the change of a specific argument's final strength in a desirable way, i.e., so that a desired final strength is achieved or a final strength threshold is exceeded~\cite{yin2024qarg}.
\end{enumerate*}
Our research differs from \emph{(i)} in that we explain how to achieve a \emph{counterfactual}, desired ordering that requires searching for a corresponding graph (changing only initial strengths); in contrast, \cite{KAMPIK2023109066} defines explanations as subsets of changed arguments in a graph, given \emph{actual} (factual) changes.
In comparison to \emph{(ii)}, the (strong) counterfactual problem defined in~\cite{yin2024qarg} can be reduced to a special case of the SXs we introduce in our work.
Finally, our work is conceptually, but not formally, related to the idea of a \emph{semifactual} explanation that, in contrast to a counterfactual explanation, describes a somewhat ``maximal'' change to a graph that does not affect the outcome of an inference (either in general or to the extent required), as first introduced to argumentative XAI in~\cite{KR2024-2}.

Beyond explainability, our work adds
to the study of \emph{enforcement}, investigating how desired outcomes can be achieved (``enforced'') in different variants of computational argumentation~\cite{doutre-argument}.
Here, the most relevant work introduces the \emph{inverse problem} that
describes the assignment of initial strengths to a gradual argumentation framework such that a desirable outcome in terms of a final strength-based ordering of the arguments is achieved~\cite{DBLP:conf/ijcai/OrenYVB22}.
Our work extends the inverse problem to generate explanations, in order to support bipolar (instead of attack-only) argumentation graphs, and to allow for constraining the set of arguments whose initial strengths can be changed and whose final strengths are of interest.
We provide a more formal integration with~\cite{yin2024qarg,DBLP:conf/ijcai/OrenYVB22} by means of an analysis presented in Section~\ref{sec:characterizations}.

\section{Preliminaries}
\label{sec:preliminaries}
%
A \emph{Quantitative Bipolar Argumentation Graph (QBAG)}~\cite{Potyka:2019,Baroni:Rago:Toni:2019} is a tuple
$\graph = \QBAG$ consisting of  a set of arguments 
$\Args$, an \emph{attack} relation 
$\Att \subseteq \Args \times \Args$, a \emph{support} relation 
$\Supp \subseteq \Args \times \Args$
such that $\Att \cap \Supp = \emptyset$,
and an \emph{initial strength} function 
$\is : \Args \to \interval$.
%
We denote the class of QBAGs by $\cal Q$.
Given a QBAG $\graph = \QBAG$, 
we denote by
${\cal R}_\graph^-(\argx)$ and ${\cal R}_\graph^+(\argx)$ the sets $\{\argy\mid\argy \in \Args, (\argy, \argx) \in \Att\}$ and $\{\argy\mid\argy \in \Args, (\argy, \argx) \in \Supp\}$, respectively, which we call the attackers/supporters of $\argx$. 
We drop the subscript $\graph$ where the context is clear.

For $\argx, \argy \in \Args$, 
we say that ``$\argx$ can reach $\argy$'' iff there is a directed path from $\argx$ to $\argy$ in $(\Args, \Att \cup \Supp)$; 
for $S, S' \subseteq \Args$, we say that ``$S$ can reach $S'$'' iff $\exists \argx \in S, \argy \in S'$ such that $\argx$ can reach $\argy$; analogously, we may say that ``$\argx$ can reach $S'$'', given that $\{\argx\}$ can reach $S'$.
Given $S \subseteq \Args$, we define $\graph \downarrow_{S} := \left(S, \tau \cap (S \times \interval), \Att \cap (S \times S), \Supp \cap (S \times S) \right)$.

Gradual semantics determine the final strengths of arguments in a QBAG.
\begin{definition}[Gradual Semantics and Strength Function~\cite{Potyka:2019,Baroni:Rago:Toni:2019}]
\label{def:semantics}
A \emph{gradual semantics} $\fs$ defines for $\graph = \QBAG$ a (possibly partial) \emph{final strength function} $\fs_{\graph} : \Args \to \interval \cup \{ \bot \}$ that assigns the \emph{final strength} $\fs_{\graph}(\argx)$ to each $\argx \in \Args$, 
where $\bot$ is a reserved symbol meaning `undefined'. 
\end{definition}

There is a variety of gradual semantics \cite{baroni2015automatic,rago2016discontinuity,amgoud2018evaluation,Potyka:2018}, most of which belong to the class of \emph{modular semantics} \cite{mossakowski2018modular}.
Modular semantics define the strengths
of arguments by an iterative process. The strength values of all arguments are initialized with their initial strengths. Then the arguments' strength values are updated based on the strengths of their parents and the base score until they converge.
Since the procedure can fail to
converge in cyclic graphs, Definition~\ref{def:semantics} defines the final strength function as a partial function.

Modular semantics are called \emph{modular} because their update function can be decomposed into an aggregation function that aggregates the
strength values of attackers and supporters, and an influence function that uses the aggregate to adapt the initial strength.
Intuitively, supporters increase the aggregate while attackers decrease it based on their strengths. A positive aggregate will increase the
initial strength, while a negative aggregate will decrease it (cf. Example~\ref{ex:intro}).
Table~\ref{table:semantics} provides some aggregation and influence functions from the literature~\cite{mossakowski2018modular,Potyka:2019}.
By combining them, we can obtain the semantics listed in Table~\ref{table:semanticsExamples}
including 
DF-QuAD~\cite{rago2016discontinuity},
Euler-based~\cite{amgoud2018evaluation}
and quadratic energy~\cite{Potyka:2018} semantics.
\begin{table}[ht]
\centering
\footnotesize{
\begin{tabular}{ll}
\hline
\multicolumn{2}{c}{\textbf{Aggregation Functions}} \\ \hline
Sum & $\alpha^{\Sigma}_{v}(s) = \sum_{i = 1}^n v_i \times s_i $  \\
Product & $\alpha^{\Pi}_{v}(s) = \prod_{i:v_i=-1} (1 - s_i) - \prod_{i:v_i=1} (1 - s_i)$  \\
\hline
\multicolumn{2}{c}{\textbf{Influence Functions}} \\ \hline
Linear($k$) & $\iota^{l}_{w}(s) = w - \frac{w}{k} \times max\{0,-s\} + \frac{1-w}{k} \times max\{0, s\}$ \\
Euler-based & $\iota^{e}_{w}(s) = 1 - \frac{1-w^2}{1 + w \times e^s}$ \\
p-Max($k$) & $\iota^{p}_{w}(s) = w - w \times h(- \frac{s}{k}) + (1-w) \times h(\frac{s}{k}),$  \\
for $p \in \mathbb{N}$  & where $h(x) = \frac{max\{0,x\}^p}{1 + max\{0,x\}^p}$  \\
\end{tabular}}
\caption{Common aggregation $\alpha$ and influence $\iota$ functions. Intuitively, $s \in [0,1]^n$ is a strength vector 
(associating each argument with its current strength),
$v \in \{-1,0,1\}^n$ is a relationship vector indicating
which arguments attack ($-1$), support ($1$) or are in no 
relationship to ($0$) the argument of interest,
and $w$ is an initial strength.
}
\label{table:semantics}
\end{table}

\begin{table}[ht]
\centering
\footnotesize{
\begin{tabular}{lll}
\hline
\textbf{Semantics}           & \textbf{Aggregation} & \textbf{Influence}  \\ \hline
DF-QuAD (DFQuAD)       & Product     & Linear(1) \\
Euler-Based (EB)   & Sum         & EulerBased  \\
QuadraticEnergy (QE)        & Sum         & 2-Max(1)  
\end{tabular}}
\caption{Examples of gradual semantics.}
\label{table:semanticsExamples}
\end{table}
%

%
%
%
%
%
%
We can compare gradual semantics based on their satisfaction of argumentation \emph{principles}.
Such principles can help us find classes of cases for when, and when not, we can find strength change explanations.
Below we provide the definitions of the principles that are relevant for our work.

\emph{Directionality} describes that adding attacks or supports 
can only affect their directed successors.
\begin{principle}[Directionality]\label{sprinciple:directionality}
A gradual semantics $\fs$ satisfies \emph{directionality} iff
for all 
$\graph = (\Args, \is, \att, \supp)$,
$\graph' = (\Args, \is, \att', \supp')$, and $\argx \in \Args$ s.t. $\att \cup \supp = \att' \cup \supp' \cup (\argy, \argz)$ for some $\argy \in \Args$, $\argz \in \Args \setminus \{\argx\}$ s.t. there is no directed path from $\argz$ to 
$\argx$, we have $\fs_{\graph}(\argx) = \fs_{\graph'}(\argx)$.
\end{principle}
Let us note that directionality is a very weak assumption and is satisfied by all modular semantics.
\begin{proposition}[\cite{PB24}, Theorem 20]
Every modular semantics satisfies directionality.
\end{proposition}

Below, we assume an arbitrary QBAG $\graph = (\Args, \is, \att, \supp)$ and $\argx \in \Args$.
\emph{Stability} describes that any differences between initial and final strength of an argument depend on the presence of attackers or supporters.
\begin{principle}[Stability]\label{sprinciple:stability}
    A gradual semantics $\sigma$ satisfies the \emph{stability principle} iff ${\cal R}^-(\argx) = {\cal R}^+(\argx) = \emptyset$ implies that $\fs(\argx) = \is(\argx)$.
\end{principle}
Stability is a special case of the \emph{balance} property \cite{BaroniRT18}, which states that if the attackers and supporters of an argument
are equally strong, then its final strength is just its base score.
\begin{principle}[Balance]\label{sprinciple:balance}
    A gradual semantics $\sigma$ satisfies the \emph{balance principle} iff it holds that if the multisets $\{\fs(\argy) \mid \argy \in {\cal R}^-(\argx)\}$ and $\{\fs(\argy) \mid \argy \in {\cal R}^+(\argx)\}$ are equal then $\fs(\argx) = \is(\argx)$. 
\end{principle}
\emph{Balance} does not hold for arbitrary modular semantics, but for 
every \emph{elementary} modular semantics. 
\begin{proposition}[\cite{PB24}, Theorem 20]
Every elementary modular semantics satisfies balance.
\end{proposition}
Most QBAG semantics are elementary modular semantics including the
Euler-based~\cite{amgoud2018evaluation}, DF-QuAD~\cite{rago2016discontinuity}, Quadratic Energy~\cite{Potyka:2018}, MLP-based~\cite{Potyka21},
TAC and CRC semantics~\cite{PB24}. They therefore satisfy both directionality and balance.

\section{Strength Change Explanations (SXs)}
\label{sec:explanations}
Strength change explanations explain how to establish a desired \emph{final strength ordering} of arguments of interest, which we may call \emph{topic arguments}.
Here and henceforth, we assume a QBAG $\graph = \QBAG$ and a final strength function $\fs_{\graph}$, omitting the subscript $\graph$ where the context is clear.
\begin{definition}[Final Strength Ordering, $\preceq$-Satisfaction]
\label{def:fs-ordering}
Given $S \subseteq \Args$, we define the \emph{final strength ordering} $\preceq^S_{\graph,\fs}$
as follows:
$$\preceq^S_{\graph,\fs} := \{(\argx, \argy)\mid \argx, \argy \in S, \fs_\graph(\argx)\leq  \fs_\graph(\argy) \}.$$
For $\preceq \subseteq \Args \times \Args$ 
we say that ``$\graph$ satisfies $\preceq$ w.r.t. $\fs$'' iff $\preceq^S_{\graph,\fs} = \preceq$ for some $S \subseteq \Args$.
For $\argx, \argy \in \Args$ s.t. $(\argx, \argy) \in \preceq^S_{\graph,\fs}$ we use the short-hand $\argx \preceq^S_{\graph,\fs}~\argy$; we drop the subscript $\fs$ where the context is clear.
\end{definition}
We go back to the initial figure to give an example of a final strength ordering.
\begin{example}
\label{ex:final-strength-ordering}
Consider again $\graph$ from Figure~\ref{fig:intro:g}.
Because $\preceq^\Args_\graph$ is a total preorder, we can represent it as a sequence of sets, i.e., we can denote $\preceq^\Args_\graph$ by $\langle \{ \argd, \arge \}, \{ \arga \}, \{ \argc \}, \{ \argb \}   \rangle$.
\end{example}
Our strength change explanations will change the initial strengths of some the arguments of a QBAG.
We first define such \emph{strength changes} in general.
\begin{definition}[Strength Change]
\label{def:strength-change}
    A \emph{strength change} of a QBAG $\graph = \QBAG$ is a (possibly partial) function $\delta_\graph: \Args \rightarrow \interval$
    such that $\delta_\graph(\argx) \neq \tau(\argx)$ whenever $\delta_\graph(\argx) \neq \bot$.
    We call $\ddom(\delta_\graph) = \{\argy \mid \argy \in \Args, \delta_\graph(\argy) \neq \bot \}$
    its \emph{defined domain}.
\end{definition}
We let $\delta_\emptyset := \emptyset$, which we call the \emph{empty strength change}.

When a strength change is applied to a QBAG, it yields a QBAG with an updated initial strength function.
\begin{definition}[Strength Change Application]
\label{def:strength-change-application}
    We define the application of a strength change $\delta_\graph$ to $\graph$, denoted by $f(\graph, \delta)$, as $(\Args, \is', \Att, \Supp)$, where for $\argx \in \Args$:
    \begin{align*}
        \is'(\argx)   =
        \begin{cases}
          \delta(\argx), & \text{if}\ \argx \in \ddom(\delta); \\
          \is(\argx), & \text{otherwise}.
    \end{cases}
    \end{align*}
\end{definition}
We illustrate the concept with our running example.
\begin{example}
\label{ex:strength-change}
In Figure~\ref{fig:intro}, we can see the following strength changes and their applications:
\begin{itemize}
    \item $\delta'_\graph = \{(\arga, 2), (\arge, 3)\}$; applied to $\graph$ we obtain $f(\graph, \delta') = \graph'$.
    \item $\delta''_\graph = \{(\arga, 3)\}$; applied to $\graph$ we obtain $f(\graph, \delta'') = \graph''$.
    \item $\delta^*_\graph = \{(\arga, 2), (\arge, 4)\}$; applied to $\graph$ we obtain $f(\graph, \delta^*) = \graph^*$.
\end{itemize}
\end{example}
%

Now we can define which strength changes amount to SXs.
For this, we assume a preorder $\preceq$ on our set of arguments $\Args$, representing the desired final strength ordering.
For $\argx, \argy \in \Args$ s.t. $(\argx, \argy) \in \preceq$, we use the shorthand $\argx \preceq \argy$, and if and only if it then does not hold that $\argy \preceq \argx$, we may use  $\argx \prec \argy$.

\begin{definition}[Final Strength Change Explanation (SX)]
\label{def:fscx}
     A strength change $\delta_\graph$ of $\graph$ is a \emph{final Strength change eXplanation} (SX) of $\preceq$ w.r.t. $\fs$ and the \emph{mutable argument set} $M \subseteq \Args$ iff it holds that $\ddom(\delta_\graph) \subseteq M$ and $f(\graph,\delta)$ satisfies $\preceq$ w.r.t. $\sigma$. We denote the set of all SXs of $\preceq$ w.r.t. $\fs$ and $M$ by $SX_{\graph,M}^{\fs}(\preceq)$. $\{\argx | \argx \in \Args, \exists \argy \in \Args: (\argx \preceq \argy)$ or $(\argy \preceq \argx)\}$ is called the \emph{topic set}, denoted by $T(\preceq)$.
\end{definition}
Whenever we have SXs w.r.t. $\sigma$ and $\Args$, we may drop the subscript $M$ and denote them by $SX_\graph^\fs(\preceq)$; analogously we may then call $\delta_\graph \in SX_\graph^\fs(\preceq)$ an ``SX of $\preceq$ w.r.t. $\fs$''.
\begin{example}
\label{ex:scx}
We continue Example~\ref{ex:strength-change} and let $\preceq$ be the reflexive closure of $\{(\argb, \argc)\}$. We can verify that $\delta'_\graph$, $\delta''_\graph$, and $\delta^*_\graph$ are all SXs of $\preceq$ w.r.t. our naive semantics $\fs$ and $M = \{\arga, \arge \}$ as we achieve the desired ordering via $\fs_{\graph'}(\argc) > \fs_{\graph'}(\argb)$, $\fs_{\graph''}(\argc) > \fs_{\graph''}(\argb)$, and $\fs_{\graph^*}(\argc) > \fs_{\graph^*}(\argb)$.
\end{example}

The empty strength change $\delta_\emptyset$ is an SX only if a given QBAG already satisfies the desired final strength ordering.
\begin{restatable}{lemma}{lemmaemptyisx}\label{lemma:empty-is-x}
    $\forall M \subseteq \Args$, it holds that $\delta_\emptyset \in SX_{\graph,M}^\fs(\preceq)$ iff $\graph$ satisfies $\preceq$ w.r.t. $\fs$.
\end{restatable}
Note that the proofs of lemmata have been relegated to the technical appendix.

\section{Optimal SXs}
\label{sec:optimality}

Strength change explanations should not modify the original QBAG more than necessary.
%
%
This desideratum gives rise to our notion of \emph{approximate} and \emph{optimal} SXs.
As a prerequisite, we define the \emph{amount of change} of an SX as the sum of its arguments' deltas to their initial strengths in the QBAG.
\begin{definition}[Amount of Change]
Given a strength change $\delta_\graph$, we call $\| \delta_\graph \| = \sum_{\argx \in \ddom(\delta_\graph)} | \delta_\graph(\argx) - \is(\argx)|$
the \emph{amount of change of $\delta_\graph$}.
\end{definition}
An SX is an $\epsilon$-approximate SX if it changes the initial strengths of arguments by not more than $\epsilon$ more than necessary. Hence, a  $0$-approximate SX is \emph{optimal}.
\begin{definition}[Approximate and optimal SX]
\label{def:optimal-fscx}
$\delta_\graph \in SX^\fs_{\graph, M}(\preceq)$ is an \emph{$\epsilon$-approximate SX} of $\preceq$ w.r.t. $\fs$, and $M \subseteq \Args$ iff there exists no $\delta'_\graph \in SX^\fs_{\graph, M}(\preceq)$ s.t. $\|\delta_\graph' \| < \|\delta_\graph\|- \epsilon$ (with $\epsilon \in \mathbb{R}_{\geq 0}$).
$\delta_\graph$ is an \emph{optimal SX} iff it is a $0$-approximate SX.
We denote the set of all $\epsilon$-approximate SXs of $\preceq$ w.r.t. $\fs$ and $M$ by $SX_{\graph, M}^{\fs*}(\preceq, \epsilon)$.
\end{definition}
Again, whenever we have $\epsilon$-approximate SXs w.r.t. $\sigma$ and $\Args$, we may drop the subscript $M$ and simply denote them by $SX_\graph^{\fs*}(\preceq, \epsilon)$; analogously we may then call $\delta_\graph \in SX_\graph^{\fs*}(\preceq,\epsilon)$ an ``$\epsilon$-approximate SX of $\preceq$ w.r.t. $\fs$''.
\begin{example}
\label{ex:optimal-scx}
Consider again the desired ordering $\preceq$ (Example~\ref{ex:scx}), the mutable set $M = \{\arga, \arge \}$, and the strength changes $\delta'_\graph$, $\delta''_\graph$, and $\delta^*_\graph$ (Example~\ref{ex:strength-change}).
We can observe that:
\begin{itemize}
    \item $\delta'_\graph$ and $\delta''_\graph$ are $1$-approximate SXs of $\preceq$ w.r.t. our naive semantics $\fs$ and $M$. In the case of $\delta'_\graph$, we change the initial strengths of $\arga$ to $2$ and of $\arge$ to $3$ (in sum, a change of $\, 2$). Given our wiggle room of $\epsilon = 1$, we cannot change the initial strengths of $\arga$ and $\arge$ substantially less to achieve the desired ordering: while we could change the initial strength of $\arga$ just marginally more and abstain from changing $\arge$, the change to $\arga$ would then amount to $> 1$, which is greater than $2 - \epsilon$. In the case of $\delta''_\graph$, we clearly need to change the initial strength of at least one argument and we cannot change the initial strength of $\arga$ substantially less (a change to at least marginally greater than $2$ is required).
    \item In contrast, $\delta^*_\graph$ is not a $1$-approximate SX, as the strength change $(\{\arga, 2.99\})$ (a slightly smaller change than $\delta''_\graph$) achieves the desired ordering by changing the initial strength of $\arga$ substantially (by more than $1$) less.
\end{itemize}
\end{example}
If a desired strength ordering is already satisfied, only the empty strength change is an optimal SX.
\begin{restatable}{lemma}{lemmaemptyisoptimal}\label{lemma:empty-is-optimal}
    If $\graph$ satisfies $\preceq$ w.r.t. $\fs$ then $SX_\graph^{\fs*}(\preceq, 0) = \{\delta_\emptyset\}$.
\end{restatable}
%

\section{(Non-)Existence of SXs}
\label{sec:approximating}
%
Finding SXs is a difficult problem.
Note that SXs may not exist.
For example, if the mutable arguments cannot reach our topic arguments, given modular semantics we cannot achieve 
a desired ordering by changing their initial strengths.
Below, we analyse some basic properties w.r.t. the existence and non-existence of SXs.
As a prerequisite, we introduce some additional gradual semantics principles.
The first one is a variant of directionality (Principle~\ref{sprinciple:directionality}) pertaining to the existence of arguments rather than edges.
\begin{principle}[Strong Directionality]\label{sprinciple:strong-directionality}
A gradual semantics $\fs$ satisfies the \emph{strong directionality principle} iff
for all 
$\graph = (\Args, \is, \att, \supp)$, $\argx \in \Args$, and $\graph' := \graph \downarrow_{\Args \setminus \Args'}$ s.t. $\Args'$ cannot reach $\argx$ it holds that $\fs_\graph(\argx)  = \fs_{\graph'}(\argx)$.   
\end{principle}
Modular semantics satisfy strong directionality.
\begin{restatable}{lemma}{lemmastrongdirectionality}\label{lemma:strong-directionality}
    Every modular semantics satisfies strong directionality.
\end{restatable}

Next, we consider a variant of monotonicity~\cite{BaroniRT18} that explicitly assumes an initial strength difference
(the general principle is not sufficiently explicit about this case for our purposes).
\begin{principle}[Weak Monotonicity]
\label{sprinciple:smonotonicity}
A gradual semantics $\fs$ satisfies the \emph{weak monotonicity principle} iff
for all 
$\graph = (\Args, \is, \att, \supp)$, $\argx, \argy \in \Args$, the following statements hold if ${\cal R}^-(\argx) \supseteq {\cal R}^-(\argy)$ and ${\cal R}^+(\argx) \subseteq {\cal R}^+(\argy)$:
\begin{enumerate}
    \item if $\is(\argx) \leq \is(\argy)$ then $\fs(\argx) \leq \fs(\argy)$;
    \item if $\sigma(\argy) < \sigma(\argx)$ then $\is(\argy) < \is(\argx)$.
\end{enumerate}
\end{principle}
Intuitively, one would expect that many modular semantics satisfy weak monotonicity: initially weaker arguments with strictly less (or the same) attackers and more (or the same) supporters should be finally weaker as well.
\begin{proposition}
    DFQuAD, EB, and QE semantics satisfy weak monotonicity.
\end{proposition}
\begin{proof}
   Consider an aggregation function $\alpha$ that is either \emph{Product}, as applied by DFQuAD semantics, or \emph{Sum}, as applied by EB and QE semantics (cf. Tables~\ref{table:semantics} and~\ref{table:semanticsExamples}).
   Given two arguments $\argx$ and $\argy$ s.t. ${\cal R}^-(\argx) \supseteq {\cal R}^-(\argy)$, and ${\cal R}^+(\argx) \subseteq {\cal R}^+(\argy)$ and their strength and relationship vectors $s_\argx$, $v_\argx$ and $s_\argy$, $v_\argy$, respectively, it holds that $\alpha_{v_\argx}(s_\argx) \leq \alpha_{v_\argy}(s_\argy)$.
   Given this and if $\is(\argx) \leq \is(\argy)$ (Principle~\ref{sprinciple:smonotonicity}, Condition 1), it follows for an influence function $\iota$ that is either \emph{Linear(1)} (for DFQuAD semantics),  \emph{EulerBased} (for EB), or \emph{2-Max(l)} (for QE) that $\iota_{\is(\argx)}(\alpha_{v_\argx}(s_\argx)) \leq \iota_{\is(\argy)}(\alpha_{v_\argy}(s_\argy))$; 
   conversely, because $\alpha_{v_\argx}(s_\argx) \leq \alpha_{v_\argy}(s_\argy)$, if $\sigma(\argy) < \sigma(\argx)$ (Condition 2), this can only be achieved by differences in initial strengths, i.e., given $\is(\argy) < \is(\argx)$.
\end{proof}

Below, we assume our final strength function $\fs$ is based on a modular gradual semantics and we only consider QBAGs that do not have undefined final strengths given $\fs$.
We first give several conditions under which we cannot find SXs, given the desired ordering $\preceq$ is currently not satisfied, i.e., we assume that $\graph$ does not satisfy $\preceq$ w.r.t. $\fs$. 

If two arguments whose relative final strengths need to change cannot be reached by the set of mutable arguments,
then we cannot find an SX.
\begin{proposition}
\label{prop:not-reachable}
Given a modular semantics $\fs$ it holds that $SX_{\graph,M}^\fs(\preceq) = \emptyset$ if $\exists \argx, \argy \in \Args$ s.t. $\argx \preceq \argy$ but  $ \argx \not \preceq^\Args_\graph~\argy$ and $M$ cannot reach $\{\argx, \argy\}$.
\end{proposition}
\begin{proof}
    Consider $\argx, \argy \in \Args$ s.t. $\argx \preceq \argy$ but $\argx \not \preceq^\Args_\graph~\argy$ (as assumed by the proposition).
    Observe that every modular semantics $\fs$ satisfies strong directionality (Lemma~\ref{lemma:strong-directionality}).
    Thus, because $M$ cannot reach $\{\argx, \argy\}$, for every $\delta_\graph$ s.t. $\{\argz\mid(\argz, s) \in \delta_\graph\} \subseteq M$, for every $S \subseteq \Args$ it must hold that $\argx \not \preceq^S_{f(\graph,\delta)}~\argy$ and consequently $\delta_\graph \not \in SX_{\graph, M}^{\fs}(\preceq)$, whence $SX_{\graph, M}^{\fs}(\preceq) = \emptyset$.
\end{proof}
An immediate consequence is
that we cannot find an SX if no mutable argument can reach any of the topic arguments.
\begin{restatable}{corollary}{nosxnotreachable}
 Given a modular semantics $\fs$, $SX_{\graph, M}^{\fs}(\preceq) = \emptyset$ holds if $M$ cannot reach $T(\preceq)$ and $\graph$ does not satisfy $\preceq$ w.r.t. $\fs$.
\end{restatable}
\begin{proof}
    As $\graph$ does not satisfy $\preceq$ w.r.t. $\fs$ it must hold that $\exists (\argx, \argy) \in \Args$ s.t.  $\argx \preceq \argy$ but $M$ cannot reach $\{\argx, \argy\}$ and for every $S \subseteq \Args$ it holds that $\argx \not \preceq^S_\graph~\argy$. Hence, the proof follows directly from Proposition~\ref{prop:not-reachable}.
\end{proof}

Assuming our gradual semantics satisfies weak monotonicity, we cannot find an SX, either, if two topic arguments whose relative final strengths need to \emph{inverse} are not mutable arguments and have the same attackers and supporters.
\begin{proposition}
    Given a semantics $\fs$ that satisfies weak monotonicity it holds that $SX_{\graph, M}^{\fs}(\preceq) = \emptyset$ if $\exists \argx, \argy \in \Args$ s.t.  $\argx \prec \argy$ but $\argx \not \preceq^\Args_\graph~\argy$, $\argx, \argy \not \in M$, and ${\cal R}^-(\argx) = {\cal R}^-(\argy)$, as well as ${\cal R}^+(\argx) = {\cal R}^+(\argy)$.
\end{proposition}
\begin{proof}
    Consider $\argx, \argy \in \Args$ s.t. $\argx \preceq \argy$ but $\argx \not \preceq^\Args_\graph~\argy$, as well as $\argx, \argy \not \in M$ and ${\cal R}^-(\argx) = {\cal R}^-(\argy)$, as well as ${\cal R}^+(\argx) = {\cal R}^+(\argy)$ (as assumed in the proposition).
    This means that $\fs(\argx) > \fs(\argy)$ must hold (as implied by $\argx \not \preceq^\Args_\graph~\argy$).
    Because $\fs$ satisfies weak monotonicity, it must hold that $\is(\argx) > \is(\argy)$ (as $\argx$ and $\argy$ share all attackers and supporters).
    Consequently, for any $\delta_\graph$ s.t. $\{\argz\mid(\argz, s) \in \delta_\graph\} \subseteq M$ it must hold for $\graph' := f(\graph, \delta)$ that $\fs_{\graph'}(\argx) \geq \fs_{\graph'}(\argy)$ (otherwise, we would again violate weak monotonicity).
    Therefore, it holds that $\delta_\graph \not \in SX_{\graph, M}^{\fs}(\preceq)$ and thus $SX_{\graph, M}^{\fs}(\preceq) = \emptyset$.
\end{proof}

We now move on to some cases where we can guarantee that SXs exist.
First, if all topic arguments are mutable and have neither attackers nor supporters, we can achieve the desired ordering by
modifying their initial strengths directly, assuming our semantics satisfies stability.
\begin{proposition}\label{prop:topic-arguments-mutable}
Given a gradual semantics $\fs$ satisfying stability, \\ $SX_{\graph,T}^{\fs}(\preceq) \neq \emptyset$ if $\forall \argx \in T(\preceq)$ it holds that ${\cal R}^-(\argx) = {\cal R}^+(\argx) = \emptyset$ and $\argx \in M$.
\end{proposition}
\begin{proof}
Because $\fs$ satisfies stability and $\forall \argx \in T(\preceq)$ it holds that ${\cal R}^-(\argx) = {\cal R}^+(\argx) = \emptyset$, for every $\graph' = (\Args, \is', \Supp, \Att)$ for every initial strength function $\is'$ it must hold that $\is'(\argx) = \fs_\graph'(\argx)$.
We can hence achieve a mapping $\delta_\graph: T(\preceq) \rightarrow \mathbb{R}$ s.t. $\forall \argy,\argz \in T(\preceq)$ it holds that $\delta_\graph(\argy) \leq \delta_\graph(\argz)$ iff $\argy \preceq \argz$, thus achieving that $f(\graph, \delta)$ satisfies $\preceq$; then, by definition of an SX (Definition~\ref{def:fscx}), it must hold that $\delta_\graph \in SX_{\graph, M}^{\fs}(\preceq)$, i.e., $SX_{\graph, M}^{\fs}(\preceq) \neq \emptyset$; intuitively: as all topic arguments are mutable arguments without external influence, we can change their initial strengths directly s.t. we achieve the desired ordering $\preceq$.
\end{proof}
Similarly, we can guarantee the existence of SXs if all topic arguments are mutable arguments that cannot reach each other and we can achieve zero influence of all incoming attackers and supporters, by changing mutable arguments other than the topic arguments.
\begin{proposition}
Given a gradual semantics $\fs$ satisfying balance,  it holds that $SX_{\graph, M}^{\fs}(\preceq) \neq \emptyset$ if $T(\preceq) \subseteq M$, $\forall \argx, \argy \in T(\preceq)$ s.t. $\argx \neq \argy$ it holds that $\argx$ cannot reach $\argy$, and there exists a strength change $\delta_\graph$ s.t. $ddom(\delta_\graph) \subseteq M \setminus T(\preceq)$ and $\forall \argz \in \Args$ s.t. $\exists \argx \in T(\preceq): \argz \in {\cal R^-(\argx)} \cup {\cal R^+(\argx)}$ it holds that $\fs_{f(\delta, \graph)}(\argz) = 0$.
\end{proposition}
\begin{proof}
Because $\fs$ satisfies balance and there exists a strength change $\delta_\graph$ s.t. $\forall \argz \in \Args$ s.t. $\exists \argx \in T(\preceq): \argz \in {\cal R^-(\argx)} \cup {\cal R^+(\argx)}$ it holds that $\fs_{f(\delta, \graph)}(\argz) = 0$, for this strength change $\delta_\graph$, $\forall \argx \in T(\preceq)$ it also holds that $\fs_{f(\delta, \graph)}(\argx) = \is(\argx)$.
Because it also holds that $ddom(\delta_\graph) \subseteq M \setminus T(\preceq)$, we can apply another strength change $\delta'_\graph$ s.t. $\forall \argx \in T(\preceq)$ it still holds that $\fs_{f(\delta', f(\delta, \graph))}(\argx) = \is(\argx)$ and in addition such that any total preorder $\preceq$ (on $T(\preceq)$, obviously), can be achieved by assigning final strengths (in $\mathbb{R}$) to all arguments in $T(\preceq)$ accordingly, analogous to how we can achieve this for Proposition~\ref{prop:topic-arguments-mutable}.
\end{proof}
%

\section{Heuristic Search}
\label{sec:application}
%

\paragraph{Experimental Setups}
We conduct experiments on \emph{layered} acyclic QBAGs that we call \emph{MLP-like QBAGs} because they feature a feed-forward structure like Multi-Layer Perceptrons (MLPs).
In a layered QBAG, arguments can be partitioned into layers, such that only (and all) arguments in the first layer have no parents; arguments in the second layer then have parents only in the first layer and children in the third layer, and so forth. Only (and all) nodes in the final layer do not have children.
Layered argumentation graphs are common in applications of CA, both generally~\cite{DBLP:conf/comma/XiaOBL24,DBLP:journals/corr/abs-2409-05457,DBLP:conf/tafa/MaioS11}  and specifically for weighted argumentation variants such as QBAGs~\cite{ayoobi2024argumentative,cocarascu2019extracting,RAGO2021103506}.

While QBAGs in many real-world scenarios are naturally acyclic, extending our approach to cyclic QBAGs remains future work. Since there are no public benchmark datasets for QBAGs, we use synthetic graphs to evaluate the performance.
We next distinguish two MLP types. The first type is randomly generated and may not exhibit a solution. 
The second type is constructed with additional constraints to guarantee the existence of a solution and we refer to them as \emph{constrained QBAGs}.
For both types, we consider four different structures that vary in the number of layers and arguments per layer: $[8,32,16,3]$, $[8,32,16,8]$, $[8,64,16,8,3]$ and $[8,64,16,8,8]$. 
For example, $[8,32,16,3]$ represents 8 arguments in the 
first layer, 32 in the second, 16 in the third, and 3 in the final layer. We refer to the layers between the first and the last as \emph{intermediate layers}. Arguments are assigned random base scores uniformly sampled from $[0,1]$.
Edges are added between all arguments in adjacent layers, making the QBAGs fully connected between consecutive layers.
Each edge is independently labelled as either attack or support with equal probability. 
The topic arguments are set as those in the final layer, and the desired ordering follows the decreasing strengths of these arguments.
To reduce the effect of randomness, we create 100 QBAGs for each structure. Finally, we use DF-QuAD semantics for evaluation due to its wide applicability (cf. \cite{cocarascu2019extracting,kotonya2019gradual}).

In the constrained QBAG setting, the same structural templates are reused but with additional constraints to ensure that the SXs are guaranteed to exist.
Let the layers be denoted by $L_1$ (the fist layer) to $L_n$ (the final layer). Similarly, we set arguments in $L_n$ as the topic arguments, which are mutually independent of each other.
Our focus is on decreasing the strength of arguments in layer $L_{n-1}$ to $0$, so that they have no influence on the final layer. Then, a valid SX can be obtained if the algorithm successfully identifies a decreasing ordering of base scores in $L_n$.
To this end, we make the arguments in $L_{n-1}$ immutable and assign them with small random base scores uniformly sampled from $[0,0.1]$, so that their strength can be more easily decreased to 0 by attackers from $L_{n-2}$. 
Furthermore, we enforce that $L_{n-2}$ contains only attack relations targeting $L_{n-1}$, and $L_{n-3}$ contains only support relations targeting $L_{n-2}$. This design ensures that the strengths of arguments in $L_{n-2}$ can be maximised through supports from $L_{n-3}$, enabling them to strongly attack and minimise the arguments in $L_{n-1}$.

\paragraph{Objective Function and Optimisation Setups}
Suppose the final layer arguments are denoted by $a_1, a_2, \cdots, a_n (n>1)$ with a desired ordering $\sigma(a_1) \geq \sigma(a_2) \geq \cdots \geq \sigma(a_n)$. To find an ordering by local search, we need an objective function that decreases with respect to the number of order-constraint violations. 
We adopt the ReLU cost function 
$cost(\sigma) := \sum_{i<j}max(0,(\sigma(a_j)-\sigma(a_i)))$.

We employ the gradient descent algorithm with Adam optimiser to minimise the cost, with a maximum of $100$ iterations.
While it may converge to local minima, it can serve here as a proof of concept enabling our heuristic search.
To evaluate the optimisation results, we first check \emph{validity}, i.e., whether the final ranking exactly matches the desired ordering.
Additionally, we employ two standard ranking correlation metrics: \emph{Kendall's $\tau$} and \emph{Spearman's $\rho$} ranking correlation (cf. \cite{rankingloss2}). 
These metrics directly capture ordering quality.
Both metrics range from $-1$ (reverse order) to $1$ (equal order), with higher scores indicating better alignment with the desired ordering. 
We also report the average runtime across all QBAGs, as well as the average absolute base score difference (per argument) for those QBAGs for which we successfully identify the desired ordering.

Algorithm~\ref{algo_heuristic_search} illustrates the iterative heuristic search, which consists of three main steps. Since a valid solution may not always exist, a maximum number of iterations is set to prevent infinite loops. First, the algorithm computes the ReLU cost. If the cost equals $0$, indicating that a valid solution has been found, the algorithm returns the base score function; otherwise, it proceeds to the next step. Second, the gradients of the cost function w.r.t. each mutable argument is computed and stored. Finally, the base scores are updated based on their corresponding gradients, with the dynamic learning rate $\alpha$ provided by the Adam optimiser.
If no solution is found, we return $\texttt{null}$.

Let us formally observe the time complexity of Algorithm~\ref{algo_heuristic_search}.
\begin{restatable}{lemma}{lemmacomplexity}
\label{proposition_iterative_algo}
Let $n$ be the number of topic arguments, $K$ the maximum number of iterations, and $N=|Args|+|Att|+|Supp|$. For acyclic QBAGs, the time complexity of Algorithm~\ref{algo_heuristic_search} is $\mathcal{O}(K \cdot (|M| \cdot N + n^2))$. 
\end{restatable}

We have seen that SXs may or may not exist. Deciding their existence and finding an optimal SX are challenging problems:
Even in acyclic, layered QBAGs, the impact of one argument's initial strength on the final strength of another argument may not be monotonic. E.g., increasing an argument's initial strength by a value of, assume, 0.1, may have a positive impact on the final strength of another argument and further increasing the initial strength (e.g. by 0.11 instead of just 0.1) may then have a negative impact (cf. Figure 1 in~\cite{DBLP:journals/ijar/KampikPYCT24}).
To address this challenge pragmatically, we design and implement a local search (gradient descent) algorithm
that tries to find SX by minimising the violation of order constraints.
Strictly speaking, since our algorithm takes gradient information into account, it can only be applied if the strength
function is differentiable. While this is the case for acyclic QBAGs, the strength function for cyclic QBAGs is not
necessarily differentiable, and even if it was, it would be difficult to derive a closed-form solution for the
partial derivatives. In principle, one could replace the partial derivatives with difference quotients in this case.
Still, since the majority of QBAG applications results in acyclic graphs, we focus on this case.

\begin{algorithm}[tb]
    \caption{Heuristic Search}
    \small
    \label{algo_heuristic_search}
    \begin{flushleft}
    \textbf{Input}: QBAG $\graph = \QBAG$, semantics $\sigma$, learning rate $\alpha$, mutable set $M \subseteq \Args$, desired ordering $\sigma(\arga_1) \geq \cdots \geq \sigma(\arga_n)$ \\
    \textbf{Parameter}: Perturbation value $\varepsilon$, maximum iterations $K$ \\
    \textbf{Output}: Updated $\tau$ (which also is an SX)
    \end{flushleft}
    \begin{algorithmic}[1]
        \STATE $\nabla cost = \{  \}$ \hfill \# \textit{gradient dictionary}
        \FOR{$k = 1$ to $K$}
            \STATE \# \textit{1. compute cost}
            \STATE compute $\sigma(\arga)$ for all $\arga \in Args$ \hfill 
            \STATE $cost \leftarrow \sum_{1 \leq i<j \leq n} \max(0, \sigma(\arga_j)-\sigma(\arga_i))$ 
            \IF{$cost = 0$} 
                \STATE \textbf{return} $\tau$ \hfill \# \textit{solution found}
            \ENDIF
            \STATE \
            \STATE \# \textit{2. compute gradients}
            \FOR{$\arga$ in $M$} 
            \STATE $\tau(\arga) \leftarrow \tau(\arga) + \varepsilon$ \hfill \# \textit{perturb} $\tau(\arga)$
            \STATE compute $\sigma(\arga)$ for all $\arga \in Args$
            \STATE $cost' \leftarrow \sum_{1 \leq i<j \leq n} \max(0, \sigma(\arga_j)-\sigma(\arga_i))$
            \STATE $\nabla cost[\arga] \leftarrow (cost'-cost)/\varepsilon$
            \STATE $\tau(\arga) \leftarrow \tau(\arga) - \varepsilon$ \hfill \# \textit{restore} $\tau(\arga)$
            \ENDFOR
            \STATE \
            \STATE \# \textit{3. update base scores}
            \FOR{$\arga$ in $M$}
                \STATE $\tau(\arga) \leftarrow \max(0, \min(1, \tau(\arga) - \alpha \cdot \nabla cost[\arga]))$
            \ENDFOR
        \ENDFOR
        \STATE \textbf{return} $\texttt{null} \ \#$ \textit{No solution found}
    \end{algorithmic}
\end{algorithm}

\paragraph{Results and Analysis}
Table~\ref{tab_results} shows the experimental results.
The third column shows the results for the constrained QBAGs. 
Our algorithm consistently finds SXs, achieving $100\%$ validity, which results in the best Kendall and Spearman correlation. The runtime
increases with both the number of intermediate layers and the number of topic arguments involved in the desired ordering. 
The average absolute base score differences are larger for those QBAGs with more topic arguments.

The remaining three columns show results of the random QBAGs under different configurations of mutable arguments. Theoretically, if all arguments are mutable, SXs always exist
by directly assigning decreasing base scores to the topic arguments and zero to all others, thereby nullifying any undesired influence. Our experimental results (shown in the last column) confirm that the algorithm reliably identifies SXs under this condition. Although the validity reaches 99\% for the final configuration, the algorithm successfully could find a solution for the previously failed case after increasing the number of iterations to $1000$.
However, in cases with only partially mutable arguments, our algorithm does not always succeed. This may be attributed to several possible factors: SXs may not exist, the number of iterations may be insufficient, or the algorithm may have converged to a local minimum, which is a known limitation of gradient-based methods.
Despite these challenges, we observe a clear trend: as the number of mutable arguments increases---from only first layer mutable, to intermediate layers mutable, and finally to all layers mutable---the validity, Kendall, and Spearman correlation scores improve consistently, which aligns with our expectation\footnote{Note: the technical appendix contains results with an additional experimental setting where both first and intermediate layers are mutable; we also report results for EB and QE semantics.}.
As for the absolute base score difference, we observe that configurations with more topic arguments require larger adjustments when the number of mutable arguments is fixed.

Our experiments demonstrate that, while the general problem is challenging, our algorithm can reliably identify SXs in some scenarios where we can guarantee the existence of SXs.
Accordingly, future research towards more applied directions, e.g. by utilising SXs for MLP debugging, can be considered promising.

\begin{table}[t]
\centering

\caption{Average validity, Kendall \& Spearman correlation, runtime (in seconds), and absolute base score difference (per argument) over 100 MLP-like QBAGs with varying structures.}
\footnotesize{
\begin{tabular}{llcccc}
\toprule
\textbf{Structure} & \textbf{Metric} 
& \makecell{\textbf{Constrained} \\ \textbf{$L_{n-1}$ fixed}} 
& \makecell{\textbf{First} \\ \textbf{mutable}} 
& \makecell{\textbf{Interm.} \\ \textbf{mutable}} 
& \makecell{\textbf{All} \\ \textbf{mutable}} \\
\midrule
\multirow{5}{*}{\textbf{[8,32,16,3]}} 
& Validity        & \textbf{100\%} & 0\%      & 83\%	 & \textbf{100\%} \\
& Kendall         & 1.00 		   & -0.24    & 0.78    & 1.00 \\
& Spearman        & 1.00 		   & -0.24    & 0.78    & 1.00 \\
& Runtime         & 0.03 		   & 1.09     & 0.28    & 0.03 \\
& $|\Delta$ BS$|$ & 0.01 		   & NA       & 0.30    & 0.15 \\

\midrule
\multirow{5}{*}{\textbf{[8,32,16,8]}} 
& Validity        & \textbf{100\%} & 0\%      & 33\%    & \textbf{100\%} \\
& Kendall   	  & 1.00 		   & -0.02    & 0.62    & 1.00 \\
& Spearman  	  & 1.00 		   & -0.03    & 0.68    & 1.00 \\
& Runtime   	  & 0.11 		   & 7.46     & 1.01    & 0.11 \\
& $|\Delta$ BS$|$ & 0.06 		   & NA       & 0.39    & 0.27 \\

\midrule
\multirow{5}{*}{\textbf{[8,64,16,8,3]}} 
& Validity        & \textbf{100\%} & 3\%      & 87\%   & \textbf{100\%} \\
& Kendall    	  & 1.00 		   & -0.19    & 0.89   & 1.00 \\
& Spearman  	  & 1.00 		   & -0.20    & 0.90   & 1.00 \\
& Runtime         & 0.08 		   & 3.57     & 0.70   & 0.08 \\
& $|\Delta$ BS$|$ & 0.01 		   & $\sim0$  & 0.08   & 0.04 \\
\midrule
\multirow{5}{*}{\textbf{[8,64,16,8,8]}} 
& Validity        & \textbf{100\%} & 0\%      & 24\%   & 99\% \\
& Kendall         & 1.00 		   & 0.02 	  & 0.54   & 0.99 \\
& Spearman        & 1.00 		   & 0.03 	  & 0.61   & 0.99 \\
& Runtime         & 0.34 		   & 3.99 	  & 3.19   & 0.40 \\
& $|\Delta$ BS$|$ & 0.03 		   & NA  	  & 0.12   & 0.10 \\
\bottomrule
\end{tabular}}
\label{tab_results}
\end{table}

\section{Relating SXs to Inverse \& Counterfactual Problems}
\label{sec:characterizations}
SXs are closely related to the \emph{inverse problem} as introduced by~\cite{DBLP:conf/ijcai/OrenYVB22}, as well as to the related \emph{strong counterfactual problem}~\cite{yin2024qarg} that is defined as a stepping stone to counterfactual explanations for QBAGs.
In this section, we will show the following:
\begin{enumerate*}[label=(\roman*)]
    \item Every solution of an inverse problem is also an SX; note that the reverse is not the case as SXs can specify specific sets of topic and mutable arguments, and can start off with arbitrary initial strengths assignments;
    \item Strong counterfactual problems and their solutions can be reduced to SXs; again the reverse is not the case, as SXs cover preferences over arbitrary many arguments in a QBAG.
\end{enumerate*}

To be able to formally integrate our explanations into the body of related work, we introduce some additional definitions, starting with the \emph{inverse problem}, that given an argumentation framework without initial strengths seeks to identify an initial strength assignment that achieves a desired final strength-based ordering of the arguments.
Note that \cite{DBLP:conf/ijcai/OrenYVB22} defines the inverse problem for attack-only instead of bipolar argumentation frameworks and semantics; for the sake of conciseness, we generalise immediately to QBAGs.
\begin{definition}[Inverse Problem]
    An \emph{inverse problem} with respect to a gradual semantics $\fs$ is a 4-tuple $I = (\Args, \Att, \Supp, \preceq)$, where $\Args$ is a set of arguments, $\Att, \Supp \subseteq \Args \times \Args$, and $\preceq \subseteq \Args \times \Args$. $\preceq$ is called the \emph{desired ordering}.
    A solution of the inverse problem $I$ is an initial strength function $\is: \Args \rightarrow \interval$ s.t. $\{(\argx, \argy) \mid \argx, \argy \in \Args, \fs(\argx) \leq \fs(\argy) \} = \preceq$.
\end{definition}
We denote the class of inverse problems by $\cal I$.

A somewhat similar problem has been introduced as a prerequisite of an argumentation-based XAI approach. The \emph{strong counterfactual problem} describes, given a QBAG and a topic argument of that QBAG, the identification of an initial strength function that achieves a specific desired final strength of the topic argument~\cite{yin2024qarg}.
\begin{definition}[Strong Counterfactual Problem]
    The strong counterfactual problem with respect to an argumentation semantics $\fs$ is a 3-tuple $C = (\graph, \argx, s)$, where $\graph = \QBAG$ is a QBAG, $\argx \in \Args$, $s \in \interval$ as well as $s \neq \fs_{\graph}(\argx)$.
    The solution of the strong counterfactual problem $C$ is an initial strength function $\is' \neq \is$ such that, given $\graph' = (\Args, \is', \Att, \Supp)$, it holds that $\fs_{\graph'}(\argx) = s$.
\end{definition}
The following example illustrates the two \emph{problems}.
\begin{example}\label{ex:problems}
Consider $\graph$ in Figure~\ref{fig:inverse-problem}\footnote{Here, $\graph$ is technically not a QBAG.}, with arguments $\Args = \{\arga, \argb, \argc, \argd, \arge\}$, $\Att = \{(\arga, \argb), (\argd, \arge)\}$, and $\Supp = \{(\arga, \argc), (\argd, \arga), (\arge, \argc)\}$. With the sequence $\preceq^* = \langle \argd, \arge, \arga, \argb, \argc \rangle$ 
giving rise to the corresponding total order $\preceq$\footnote{I.e., $\preceq$ is the transitive and reflexive closure of $\{(\argd, \arge), (\arge, \arga), (\arga, \argb), (\argb, \argc)\}$.
}, we have the inverse problem $I = (\Args, \Att, \Supp, \preceq)$.
Given $\graph^*$ in Figure~\ref{fig:counterfactual-problem}, $C = (G^*, \argc, 6)$ is a strong counterfactual problem.
The initial strength function seen in $\graph'$ (Figure~\ref{fig:inverse-solution}), i.e., $\is' = \{(\arga, 2), (\argb, 8), (\argc, 1), (\argd,1 ), (\arge, 3)\}$, is a solution of $I$, and of $C$.

\vspace{-10pt}
\begin{figure}[!ht]
    \subfloat[$\graph$]{
    \label{fig:inverse-problem}
    \begin{tikzpicture}[scale=0.55]
        \node[unanode]    (a)    at(0,2)  {\pargnode{\arga}};
        \node[unanode]  (b)    at(0,0)  {\pargnode{\argb}};
        \node[unanode]    (c)    at(2,0)  {\pargnode{\argc}};
        \node[unanode]    (d)    at(1,4)  {\pargnode{\argd}};
        \node[unanode]    (e)    at(3,2)  {\pargnode{\arge}};
        
        \path [->, line width=0.5mm]  (a) edge node[left] {-} (b);
        \path [->, line width=0.5mm]  (a) edge node[above] {+} (c);
        \path [->, line width=0.5mm]  (e) edge node[right] {+} (c);
        \path [->, line width=0.5mm]  (d) edge node[left] {+} (a);
        \path [->, line width=0.5mm]  (d) edge node[right] {-} (e);
    \end{tikzpicture}
    }
    \centering
    \subfloat[$\graph^*$]{
    \label{fig:counterfactual-problem}
    \begin{tikzpicture}[scale=0.55]
        \node[unanode]    (a)    at(0,2)  {\argnode{\arga}{1}{2}};
        \node[unanode]  (b)    at(0,0)  {\argnode{\argb}{8}{6}};
        \node[unanode]    (c)    at(2,0)  {\argnode{\argc}{1}{4}};
        \node[unanode]    (d)    at(1,4)  {\argnode{\argd}{1}{1}};
        \node[unanode]    (e)    at(3,2)  {\argnode{\arge}{2}{1}};
        
        \path [->, line width=0.5mm]  (a) edge node[left] {-} (b);
        \path [->, line width=0.5mm]  (a) edge node[above] {+} (c);
        \path [->, line width=0.5mm]  (e) edge node[right] {+} (c);
        \path [->, line width=0.5mm]  (d) edge node[left] {+} (a);
        \path [->, line width=0.5mm]  (d) edge node[right] {-} (e);
    \end{tikzpicture}
    }
    \centering
    \subfloat[$\graph'$]{
    \label{fig:inverse-solution}
    \begin{tikzpicture}[scale=0.55]
        \node[unanode]    (a)    at(0,2)  {\argnode{\arga}{2}{3}};
        \node[unanode]  (b)    at(0,0)  {\argnode{\argb}{8}{5}};
        \node[unanode]    (c)    at(2,0)  {\argnode{\argc}{1}{6}};
        \node[unanode]    (d)    at(1,4)  {\argnode{\argd}{1}{1}};
        \node[unanode]    (e)    at(3,2)  {\argnode{\arge}{3}{2}};
        
        \path [->, line width=0.5mm]  (a) edge node[left] {-} (b);
        \path [->, line width=0.5mm]  (a) edge node[above] {+} (c);
        \path [->, line width=0.5mm]  (e) edge node[right] {+} (c);
        \path [->, line width=0.5mm]  (d) edge node[left] {+} (a);
        \path [->, line width=0.5mm]  (d) edge node[right] {-} (e);
    \end{tikzpicture}
    }
\caption{Inverse and strong counterfactual problems ($\graph$, with desired total order $\langle \argd, \arge, \arga, \argb, \argc \rangle$, and $\graph^*$, with topic $\argc$ and desired strength $6$, respectively) and their solutions (in $\graph'$).} 
\label{fig:problems}
\end{figure}
\end{example}

To show that SXs generalise inverse problems and their solutions, we first introduce a function that assigns an arbitrary value $s \in \interval$ as the initial strength to all arguments of an inverse problem.
\begin{definition}[Initial Strength Assignment Function]
\label{def:inverse-problem-semantics}
    The initial strength assignment function $\phi_s: {\cal I} \rightarrow {\cal Q}$, with $s \in \interval$, takes an inverse problem $(\Args, \Att, \Supp, \preceq) \in \cal I$ and returns a QBAG $(\Args, \tau, \Att, \Supp) \in {\cal Q}$ s.t. $\tau = \{ (\argx, s) |\argx \in \Args \}$.
\end{definition}
We can then show that the solution of an inverse problem is also an SX, assuming an inverse problem that is augmented with an initial strength assignment function assigning arbitrary initial strengths $s$ to all arguments, and excluding initial strength assignments of arguments to $s$ from the solution.
\begin{proposition}
    For every inverse problem $I = (\Args, \Att, \Supp, \preceq)$, for every initial strength function $\is$ that is a solution of $I$, for every $s \in \interval$ it holds that 
    $\is \setminus \{\argx | \argx \in \Args, (\argx, s) \in \is \} \in SX_{f(\phi_s(I), \is),\Args}^{\fs}$.
\end{proposition}
\begin{proof}
    By definition of an inverse problem and its solution,$f(\phi_s(I), \is)$ satisfies $\preceq$. However, for some $(\argx, s') \in \is$ it may hold that $s' = s$ and therefore $(\argx, s')$ must not occur in a strength change. Hence, $\is \setminus \{\argx | \argx \in \Args, (\argx, s) \in \is \} \in SX_{\phi_s(I),\Args}^{\fs}(\preceq)$.
\end{proof}
Similarly, we can show that strong counterfactual problems and their solutions can be reduced to SXs: a change to the initial strengths of arguments in a QBAG that leads to a desired final strength of a specific topic argument can be characterised by an SX, given we add a ``dummy argument'' to the QBAG that serves as a reference to the desired final strength of the topic argument.
Here, we assume the gradual semantics satisfies the stability principle, which we claim is a common-sense desideratum.
\begin{proposition}
    Given a gradual semantics $\fs$ satisfying stability, for every strong counterfactual problem $C = (\graph = \QBAG, \argx, s)$, for every $\is'$ that is a solution of $C$ it holds that $\is' \in SX_{\graph_\argy,\Args}^{\fs}(\preceq)$, where $\graph_\argy = (\Args \cup \{\argy\}, \is \cup \{(\argy, s)\}, \Att, \Supp)$, $\argy \not \in \Args$, and $\preceq = \{(\argx, \argy), (\argy, \argx)\}$.
\end{proposition}
\begin{proof}
    Because $\fs$ satisfies stability and $\argy \not \in \Args$, for $\graph'_\argy = (\Args \cup \{\argy\}, \is' \cup \{(\argy, s)\}, \Att, \Supp)$ it must hold that $\fs_{\graph_\argy}(\argy) = \fs_{\graph'_\argy}(\argy)  = s$ (note that $\argy$ has neither attackers nor supporters).
    This means by definition of a strong counterfactual problem and its solution, we must have $\fs_{\graph'_\argy}(\argx) = \fs_{\graph'_\argy}(\argy)$. It follows that because $\graph'_\argy = f(\graph_\argy, \is')$, it holds that $\is' \in SX_{\graph_\argy,\Args}^{\fs}(\preceq)$, with $\preceq = \{(\argx, \argy), (\argy, \argx)\}$, as achieved by $\fs_{\graph'_\argy}$. 
\end{proof}
Let us expand on Example~\ref{ex:problems} to give an intuition of the results.
\begin{example}
\label{ex:sc-inverse-counterfactual}
Consider the previous inverse and strong counterfactual problems $I = (\Args, \Att, \Supp, \preceq)$ and $C = (\graph^* = (\Args, \is^*, \Att, \Supp), \argc, 6)$, respectively (cf. Figure~\ref{fig:problems}), as well as their solution $\is'$. We observe that:
\begin{itemize}
    \item Given $\graph_0 = (\Args, \{(\argx, 0)\mid\argx \in \Args \}, \Att, \Supp)$ it holds that $\is' \in SX_{\graph_0,\Args}^{\fs}(\preceq)$;
    \item Given $\graph_\argy = (\Args \cup \{\argy\}, \is_\argy = \is \cup \{(\argy, 6)\}, \Att, \Supp)$ it holds that $\is' \setminus \is^* \in SX_{\graph_\argy,\Args}^{\fs}(\{(\argc, \argy), (\argy, \argc)\})$.
\end{itemize}
\end{example}
%

\section{Conclusions}
\label{sec:conclusions}
We have introduced argumentative strength change explanations, as a potential foundation for argumentative XAI and contestable AI.
Our explanations generalise solutions of previously studied problems in gradual argumentation.
We have demonstrated some (non)existence results, as well as the empirical feasibility of finding explanations in relatively small, layered QBAGs, with some expected limitations.
Future research may expand our investigations regarding theoretical existence and empirical find-ability of our strength change explanations, especially in large QBAGs, measure other characteristics of the explanations, such as simplicity and robustness, and apply the explanations to real-world contestability problems and datasets.



\begin{acks}
This work was partially supported by the Wallenberg AI, Autonomous Systems and
Software Program (WASP) funded by the Knut and Alice Wallenberg Foundation.
\end{acks}



\bibliographystyle{ACM-Reference-Format} 
\bibliography{sample}


\newpage
\twocolumn[\section*{Appendix: Strength Change Explanations in Quantitative Argumentation}]
\subsection*{Appendix 1: Proofs}
The appendix re-states all lemmata and provides their proofs.

\lemmaemptyisx*
\begin{proof} \phantom{e\\}
\textbf{``If'' direction:} $\forall M \subseteq \Args$ it holds that if $\graph$ satisfies $\preceq$ w.r.t. $\fs$ then $\delta_\emptyset \in SX_{\graph, M}^\fs(\preceq)$.
    By definition of a strength change application (Definition~\ref{def:strength-change-application}) it holds that $f(\graph, \delta_\emptyset) = \graph$.
    Hence, if $\graph$ satisfies $\preceq$ w.r.t. $\fs$ then it must also hold that $f(\graph, \delta_\emptyset)$ satisfies $\preceq$ w.r.t. $\fs$, which implies $\delta_\emptyset \in SX_{\graph, M}^\fs(\preceq)$, by Definition~\ref{def:fscx}.
    
\textbf{``Only if'' direction:} $\forall M \subseteq \Args$ it holds that if $\graph$ does not satisfy $\preceq$ w.r.t. $\fs$ then $\delta_\emptyset \not \in SX_\graph^\fs(\preceq)$.
Because $f(\graph,\delta_\emptyset) = \graph$ it holds that if $\graph$ does not satisfy $\preceq$ w.r.t. $\fs$, $f(\graph,\delta_\emptyset)$ does not satisfy $\preceq$ w.r.t. $\fs$, either, and thus cannot be an SX.
\end{proof}

\lemmaemptyisoptimal*
\begin{proof}
  From Lemma~\ref{lemma:empty-is-x} it follows directly that if $\graph$ satisfies $\preceq$ w.r.t. $\fs$ then $\delta_\emptyset \in SX_\graph^\fs(\preceq, 0)$.
  Now, it remains to be shown that:
    \begin{enumerate}[label=(\roman*)]
        \item $\nexists \delta_\graph \in SX^\fs_{\graph, M}(\preceq)$ s.t. $\|\delta_\graph \| < \|\delta_\emptyset\|$ ($\delta_\emptyset$ is optimal);
        \item $\forall \delta_\graph \in SX_\graph^\fs(\preceq) \text{ s.t. } \delta_\graph \neq \delta_\emptyset$ it holds that $\|\delta_\graph \| > \|\delta_\emptyset\|$  (no other SX is optimal).
    \end{enumerate}
  Clearly \emph{i)} and \emph{ii)} hold because for every $\delta_\graph \in SX_\graph^\fs(\preceq)$ s.t. $\delta_\graph \neq \delta_\emptyset$ it holds that $\exists \argx \in \ddom(\delta_\graph)$ s.t. $|\delta_\graph(\argx) - \tau(\argx)| > 0$ (thus $\|\delta_\graph\| > 0$)   but $\nexists \argy \in \ddom(\delta_\emptyset)$ s.t. $|\delta_\emptyset(\argy) - \tau(\argy)| > 0$ (thus $\|\delta_\emptyset\| = 0$); hence $\|\delta_\graph \| > \|\delta_\emptyset\|$.
\end{proof}

\lemmastrongdirectionality*
\begin{proof}
The claim follows from observing that, under modular semantics, the final strength of an argument only depends on the initial strength of this argument and on the final strengths of its attackers and supporters. 
\end{proof}

\lemmacomplexity*
\begin{proof}
We first analyse the time complexity of computing the strength values of arguments. For acyclic QBAGs, these values can be computed in linear time $\mathcal{O}(N)$~\cite[Proposition 3.1]{Potyka:2019}. Since the strength values must be recomputed for each base score perturbation and there are $|M|$ mutable arguments to be perturbed, the time complexity for this step is $\mathcal{O}(|M| \cdot N)$.
Next, the time complexity of computing the cost function is $n^2$, as it involves pairwise comparisons among the $n$ topic arguments.
Finally, since the algorithm requires at most $K$ iterations, the overall time complexity of Algorithm~\ref{algo_heuristic_search} is $\mathcal{O}(K \cdot (|M| \cdot N + n^2))$.
\end{proof}
\newpage
\subsection*{Appendix 2: Full Experimental Results}
Table~\ref{tab_results_complete} contains all experimental results evaluating the heuristic search.
The experiments were run on a machine featuring an Apple M4 with 10 cores and 24 GB RAM.
The code for our experiments is available at \url{https://github.com/nicopotyka/Uncertainpy/blob/master/examples/gradual/strength_change_explanations.ipynb}.

In addition to the results in the main paper, we present results for an experimental setting where both first and intermediate layers are mutable.
Note that the results can be interpreted as an exception to the trend that adding more layers of arguments to the mutable set yields better results.
Indeed, the results are very similar to the setting where only the intermediate layers are mutable, presumably because the expansion of the search space to the first layer has little effect.
Accordingly, we consider the results as unsurprising and not contradictory to the bigger picture.

Also, additional experimental results for EB and QE semantics are provided.
The results for these semantics are broadly speaking similar to the ones for DFQuAD semantics, with close to perfect performance when a solution can be guaranteed and mixed results when it cannot.
Notable differences can be observed for some of the settings that fall into the latter class:
\begin{enumerate*}[label=(\roman*)]
    \item The search performs somewhat better for EB and (even more so for) QE semantics given only the first layer is mutable and there are only three topic arguments.
    \item The search performs worse for EB and QE semantics when the intermediate layer is mutable (no matter whether the first layer is mutable or not), except for the smallest QBAG structure ($[8,32,16,3]$).
\end{enumerate*}
These difference indicate that more comprehensive experiments may be interesting future work.

\begin{table*}[t]
\centering
\caption{Heuristic search results, with additional experiment where both first and intermediate layers are mutable, and for DFQuAD, as well as EB and QE semantics; average validity, Kendall correlation, Spearman correlation, runtime (in seconds), and absolute base score difference (per argument) over 100 MLP-like QBAGs with varying structures.}
\footnotesize{
\begin{tabular}{llccccc}
\toprule
\textbf{Structure} & \textbf{Metric} 
& \makecell{\textbf{Constrained} \\ \textbf{$L_{n-1}$ fixed}} 
& \makecell{\textbf{First} \\ \textbf{mutable}} 
& \makecell{\textbf{Interm.} \\ \textbf{mutable}} 
& \makecell{\textbf{First+Interm.} \\ \textbf{mutable}} 
& \makecell{\textbf{All} \\ \textbf{mutable}} \\
\midrule
\multicolumn{7}{c}{\textbf{DFQuAD Semantics}} \\
\midrule
\multirow{5}{*}{\textbf{[8,32,16,3]}} 
& Validity        & \textbf{100\%} & 0\%      & 83\%   & 82\% 	 & \textbf{100\%} \\
& Kendall         & 1.00 		   & -0.24    & 0.78   & 0.82    & 1.00 \\
& Spearman        & 1.00 		   & -0.24    & 0.78   & 0.82    & 1.00 \\
& Runtime         & 0.03 		   & 1.09     & 0.28   & 0.30    & 0.03 \\
& $|\Delta$ BS$|$ & 0.01 		   & NA       & 0.30   & 0.27    & 0.15 \\

\midrule
\multirow{5}{*}{\textbf{[8,32,16,8]}} 
& Validity        & \textbf{100\%} & 0\%      & 33\%    & 32\%   & \textbf{100\%} \\
& Kendall   	  & 1.00 		   & -0.02    & 0.62    & 0.57   & 1.00 \\
& Spearman  	  & 1.00 		   & -0.03    & 0.68    & 0.62   & 1.00 \\
& Runtime   	  & 0.11 		   & 7.46     & 1.01    & 1.04   & 0.11 \\
& $|\Delta$ BS$|$ & 0.06 		   & NA       & 0.39    & 0.34   & 0.27 \\

\midrule
\multirow{5}{*}{\textbf{[8,64,16,8,3]}} 
& Validity        & \textbf{100\%} & 3\%      & 87\%   & 86\%   & \textbf{100\%} \\
& Kendall    	  & 1.00 		   & -0.19    & 0.89   & 0.85   & 1.00 \\
& Spearman  	  & 1.00 		   & -0.20    & 0.90   & 0.85   & 1.00 \\
& Runtime         & 0.08 		   & 3.57     & 0.70   & 0.80   & 0.08 \\
& $|\Delta$ BS$|$ & 0.01 		   & $\sim0$  & 0.08   & 0.07   & 0.04 \\
\midrule
\multirow{5}{*}{\textbf{[8,64,16,8,8]}} 
& Validity        & \textbf{100\%} & 0\%      & 24\%   & 19\%   & 99\% \\
& Kendall         & 1.00 		   & 0.02 	  & 0.54   & 0.48 	& 0.99 \\
& Spearman        & 1.00 		   & 0.03 	  & 0.61   & 0.54 	& 0.99 \\
& Runtime         & 0.34 		   & 3.99 	  & 3.19   & 3.37 	& 0.40 \\
& $|\Delta$ BS$|$ & 0.03 		   & NA  	  & 0.12   & 0.12   & 0.10 \\
\midrule
\multicolumn{7}{c}{\textbf{EB Semantics}} \\
\midrule
\multirow{5}{*}{\textbf{[8,32,16,3]}} 
& Validity        & \textbf{100\%} & 14\%     & 88\%   & 90\% 	 & \textbf{100\%} \\
& Kendall         & 1.00 		   & -0.05    & 0.85   & 0.89    & 1.00 \\
& Spearman        & 1.00 		   & -0.06    & 0.84   & 0.89    & 1.00 \\
& Runtime         & 0.03 		   & 0.91     & 0.21   & 0.21    & 0.02 \\
& $|\Delta$ BS$|$ & 0.02 		   & 0.10     & 0.20   & 0.18    & 0.13 \\

\midrule
\multirow{5}{*}{\textbf{[8,32,16,8]}} 
& Validity        & \textbf{100\%} & 0\%      & 16\%    & 16\%   & 98\% \\
& Kendall   	  & 1.00 		   & 0.03     & 0.49    & 0.53   & 0.99 \\
& Spearman  	  & 1.00 		   & 0.04     & 0.57    & 0.59   & 1.00 \\
& Runtime   	  & 0.13 		   & 9.76     & 10.63   & 11.67  & 0.21 \\
& $|\Delta$ BS$|$ & 0.06 		   & NA       & 0.35    & 0.33   & 0.29 \\

\midrule
\multirow{5}{*}{\textbf{[8,64,16,8,3]}} 
& Validity        & \textbf{100\%} & 13\%     & 64\%   & 61\%   & \textbf{100\%} \\
& Kendall    	  & 1.00 		   & -0.09    & 0.59   & 0.56   & 1.00 \\
& Spearman  	  & 1.00 		   & -0.12    & 0.60   & 0.57   & 1.00 \\
& Runtime         & 0.09 		   & 21.54    & 13.89  & 14.74  & 0.08 \\
& $|\Delta$ BS$|$ & 0.01 		   & 0.04     & 0.19   & 0.17   & 0.14 \\
\midrule
\multirow{5}{*}{\textbf{[8,64,16,8,8]}} 
& Validity        & \textbf{100\%} & 0\%      & 3\%    & 4\%    & \textbf{100\%} \\
& Kendall         & 1.00 		   & 0.02 	  & 0.25   & 0.27 	& 1.00 \\
& Spearman        & 1.00 		   & 0.02 	  & 0.31   & 0.35 	& 1.00 \\
& Runtime         & 0.34 		   & 25.04 	  & 26.29  & 36.07 	& 9.51 \\
& $|\Delta$ BS$|$ & 0.03 		   & NA  	  & 0.29   & 0.29   & 0.27 \\
\midrule
\multicolumn{7}{c}{\textbf{QE Semantics}} \\
\midrule
\multirow{5}{*}{\textbf{[8,32,16,3]}} 
& Validity        & \textbf{100\%} & 35\%     & 82\%   & 93\% 	 & \textbf{100\%} \\
& Kendall         & 1.00 		   & 0.25     & 0.85   & 0.92    & 1.00 \\
& Spearman        & 1.00 		   & 0.26     & 0.87   & 0.92    & 1.00 \\
& Runtime         & 0.03 		   & 12.62    & 0.36   & 0.22    & 0.05 \\
& $|\Delta$ BS$|$ & 0.05 		   & 0.16     & 0.23   & 0.22    & 0.18 \\

\midrule
\multirow{5}{*}{\textbf{[8,32,16,8]}} 
& Validity        & 99\% 		   & 0\%      & 22\%    & 25\%   & 98\% \\
& Kendall   	  & 1.00 		   & 0.21     & 0.62    & 0.65   & 1.00 \\
& Spearman  	  & 1.00 		   & 0.27     & 0.69    & 0.71   & 1.00 \\
& Runtime   	  & 0.17 		   & 11.79    & 2.73    & 1.28   & 0.36 \\
& $|\Delta$ BS$|$ & 0.15 		   & NA       & 0.38    & 0.36   & 0.35 \\

\midrule
\multirow{5}{*}{\textbf{[8,64,16,8,3]}} 
& Validity        & \textbf{100\%} & 30\%     & 71\%   & 78\%   & \textbf{100\%} \\
& Kendall    	  & 1.00 		   & 0.24     & 0.67   & 0.72   & 1.00 \\
& Spearman  	  & 1.00 		   & 0.26     & 0.67   & 0.74   & 1.00 \\
& Runtime         & 0.09 		   & 2.72     & 14.45  & 1.13   & 0.12 \\
& $|\Delta$ BS$|$ & 0.01 		   & 0.09  	  & 0.20   & 0.18   & 0.15 \\
\midrule
\multirow{5}{*}{\textbf{[8,64,16,8,8]}} 
& Validity        & \textbf{100\%} & 1\%      & 4\%    & 7\%    & 97\% \\
& Kendall         & 1.00 		   & 0.13 	  & 0.42   & 0.46 	& 0.99 \\
& Spearman        & 1.00 		   & 0.17 	  & 0.50   & 0.55 	& 1.00 \\
& Runtime         & 2.00 		   & 22.78 	  & 24.08  & 4.19 	& 0.82 \\
& $|\Delta$ BS$|$ & 0.04 		   & 0.38  	  & 0.28   & 0.36   & 0.30 \\
\bottomrule
\end{tabular}}
\label{tab_results_complete}
\end{table*}
\end{document}